\documentclass[aps,twocolumn]{revtex4}
\usepackage{amssymb,amsmath,amsfonts,amsthm}
\usepackage{mathbbol,bm,bbm}
\usepackage{graphicx}

\newtheorem{proposition}{Proposition}
\newtheorem{corollary}[proposition]{Corollary}
\newtheorem{lemma}[proposition]{Lemma}

\newcommand{\idty}{\mathbb{1}}
\DeclareMathOperator{\id}{id}

\DeclareMathOperator*{\tr}{Tr}
\newcommand{\<}{\langle}
\renewcommand{\>}{\rangle}
\providecommand{\abs}[1]{|#1|}
\providecommand{\norm}[1]{\Vert #1 \Vert}

\renewcommand{\c}[1]{\mathcal{#1}}

\renewcommand{\r}[1]{\mathrm{#1}}

\usepackage{latexsym}

\begin{document}

\title{Universal bounds for the Holevo quantity, coherent information \\
and the Jensen-Shannon divergence}

\author
{Wojciech~Roga$^1$, Mark~Fannes$^2$ and Karol~{\.Z}yczkowski$^{1,3}$ \\
{\normalsize\itshape {$^1$Smoluchowski Institute of Physics, Jagiellonian University, ul.~Reymonta 4, 30-059 Krak{\'o}w, Poland}} \\
{\normalsize\itshape {$^2$Instituut voor Theoretische Fysica, Universiteit Leuven, B-3001 Leuven, Belgium}} \\
{\normalsize\itshape $^{{3}}$Centrum Fizyki Teoretycznej, Polska Akademia Nauk, Al.~Lotnik{\'o}w 32/44, 02-668 Warszawa, Poland}}

\date{{April 26, 2010}}

\begin{abstract}
The Holevo quantity provides an upper bound for the mutual information between the sender of a classical message encoded in quantum carriers and the receiver. Applying the strong sub-additivity of entropy we prove that the Holevo quantity associated with an initial state and a given quantum operation represented in its Kraus form is not larger than the exchange entropy. This implies upper bounds for the coherent information and for the quantum Jensen--Shannon divergence. Restricting our attention to classical information we bound the transmission distance between any two probability distributions by the entropic distance, which is a concave function of the Hellinger distance.
\end{abstract}

\maketitle

%\noindent
%Email: \texttt{<wojciech.roga@uj.edu.pl>}, \texttt{<mark.fannes@fys.kuleuven.be>}, and \texttt{<karol@tatry.if.uj.edu.pl>} 

%\noindent
%PACS: 
%02.10.Ud (Mathematical methods in physics, Linear algebra), 
%03.67.-a (Quantum mechanics, field theories, and special relativity, Quantum information), 
%03.65.Yz (Decoherence; open systems; quantum statistical methods)

\section{Introduction}
\label{s1}

The goal of quantum information is to efficiently apply quantum resources to encode, manipulate, and transmit information. One of the key results about transmitting classical information by quantum means is the {\sl Holevo bound}: it provides an upper bound for the accessible information or, stated differently, quantum distinguishability
between initial states cannot be increased by measuring them~\cite{nielsen}.

Assume that a source emits messages written in an alphabet $X$ and that these messages are sent to a receiver using a quantum device. Each letters $a_i$ is encoded in a quantum state, i.e., in a density matrix $\rho_i$. The receiver performs a measurement on the encoded messages and obtains classical data written in an alphabet $Y$. The probability that the source emits $a_i$ is given by $q_i$ and so with high probability messages of length $n$ will belong to a set containing about $\exp\bigl( n H(X) \bigr)$ words where $H$ is the Shannon entropy: $H(X):= -\sum_i q_i \ln q_i$. By performing a general POVM the receiver reads the letter $b_i$ with a probability $p_i$. The accessible information is then the {\sl mutual information} $H(X:Y) = H(X) + H(Y) - H(X,Y)$, where $H(X,Y)$ is the Shannon entropy of the joint probability distribution. The fundamental result of Holevo gives~\cite{Holevo,Be96} an upper bound for the mutual information, independent of the measurement:
\begin{equation}
H(X:Y) \le S\Bigl( \sum_i q_i \rho_i \Bigr) - \sum_i q_i S(\rho_i).
\label{hol}
\end{equation}
Here  $S(\rho) := -\tr \rho \ln\rho$ is the von~Neumann entropy of the state $\rho$. The right-hand side of inequality~(\ref{hol})
\begin{equation}
\chi\bigl( \{q_i,\rho_i\} \bigr) \;:=\; S\Bigl( \sum_i q_i \rho_i \Bigr) - \sum_i q_i S(\rho_i)
\label{holev}
\end{equation}
is the {\sl Holevo quantity} of the ensemble $\{q_i,\rho_i\}$.

It is well-known~\cite{nielsen} that the Shannon entropy of the probability vector $\{q_i\}$ gives an upper bound for $\chi$: $\chi\bigl( \{q_i,\rho_i\} \bigr) \le H(\{q_i\})$. In this work we provide a better upper bound for the Holevo quantity and explore some of its consequences for classical and quantum information theory.

\section{Quantum Operations}
\label{s2}

Consider a quantum system $A$ in the state $\rho$ interacting with an environment $B$ initially in a pure state $|\phi\rangle \in \c H_B$. Any quantum operation $\Phi$ on the system can by seen as a global unitary dynamics followed by the partial trace over the environment:
\begin{equation}
\Phi: \rho \mapsto \rho' = {\tr}_B \Bigl( U\, \bigl( \rho\otimes|\phi\>\<\phi| \bigr)\, U^\dagger \Bigr).
\label{unit}
\end{equation}
Here $U$ is a unitary matrix of the  total system $AB$ and ${\tr}_B$ denotes the partial trace over the environment.
The map $\Phi$ is completely positive and can be represented in Kraus form:
\begin{equation}
\Phi(\rho) = \sum_i K_i \rho K_i^\dagger,
\label{quantop}
\end{equation}
where the operators $K_i$ are determined by minors of $U$. Due to the unitarity of $U$ the set of Kraus operators is
a resolution of the identity: $\sum_i K_i^\dagger K_i = \idty$ which implies that $\Phi$ preserves the trace.

We consider also a quantum map $\tilde\Phi$ {\sl complementary} to $\Phi$ defined by the partial trace over the principal system~\cite{alickifannes,Ho05}:
\begin{equation}
\tilde\Phi: \rho \mapsto \sigma = {\tr}_A \Bigl(U\, \bigl(\rho\otimes|\phi\>\<\phi|\bigr)\, U^\dagger \Bigr).
\end{equation}
The state $\sigma = \tilde\Phi(\rho)$ is the state of the environment after the interaction and is called a {\sl correlation matrix}. Its matrix elements can be expressed in terms of the Kraus operators:
\begin{equation}
\sigma_{ij} = \tr \rho K_j^\dagger K_i.
\label{sigma}
\end{equation}
If the initial state $\rho$ is pure then $S(\sigma)$ is the entropy exchanged between the system and the environment. Therefore $S(\sigma)$ is called the {\sl exchange entropy}.

\begin{figure}[htbp]
\centering
\includegraphics[width=0.45\textwidth]{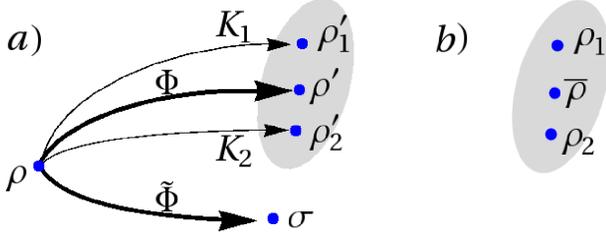}
\caption{a) A dynamical picture: an initial quantum state $\rho$ is sent by a map $\Phi$ into $\rho'$, while the complementary map $\tilde\Phi$ sends it into $\sigma$. A Kraus operator $K_i$ maps $\rho$ into $\rho'_i$ with probability $q_i$ so that $\rho'$ is the barycentre of the ensemble $\{q_i,\rho'_i\}$. \\
b) A static picture: the ensemble $\{q_i,\rho_i\}$ determines the barycentre $\bar\rho$.}
\label{fig:fig1}
\end{figure}

Due to the identity resolution the set of Kraus operators describes a Positive Operator-Valued Measure. Such a selective measurement transforms an initial state $\rho$ into one of the output states $\rho'_i := K_i \rho K_i^\dagger/ (\tr \rho K_i^\dagger K_i)$ with probability $q_i = \tr \rho K_i^\dagger K_i$. For this setup, shown  in Fig~\ref{fig:fig1}a, one defines the Holevo quantity $\chi\bigl( \{q_i,\rho'_i\} \bigr)$, see~(\ref{holev}).

\section{The Main Result}
\label{s3}

To prove our main result we first reformulate the strong sub-additivity inequality for quantum entropy.

\begin{proposition}
\label{pro1}
Let $\omega_{123}$ be a three party quantum state, then
\begin{equation}
S(\omega_1) + S(\omega_3) \le S(\omega_{12}) + S(\omega_{23}).
\label{ssa2}
\end{equation}
Conversely, if~(\ref{ssa2}) holds for any three party quantum state then also
\begin{equation}
S(\omega_{123}) + S(\omega_2) \le S(\omega_{12}) + S(\omega_{23}).
\label{ssa1}
\end{equation}
In these inequalities reduced states are obtained by tracing out over the complementary subsystems, e.g., $\omega_1 = {\tr}_{23}\, \omega_{123}$.
\end{proposition}

\begin{proof}
The proof relies on purification. Indeed, purifying $\omega_{123}$ to $\omega_{1234}$ we have the strong sub-additivity inequality
\begin{equation}
S(\omega_{234}) + S(\omega_3) \le S(\omega_{23}) + S(\omega_{34}).
\label{ssa3}
\end{equation}
As $\omega_{1234}$ is pure, $S(\omega_{234}) = S(\omega_1)$ and $S(\omega_{34}) = S(\omega_{12})$ which turns~(\ref{ssa3}) into~(\ref{ssa2}).

Conversely, if~(\ref{ssa2}) holds, we apply it to a purification $\omega_{1234}$ of $\omega_{123}$ to obtain
\begin{equation}
S(\omega_2) + S(\omega_4) \le S(\omega_{23}) + S(\omega_{24}).
\label{ssa4}
\end{equation}
Using $S(\omega_4) = S(\omega_{123})$ and $S(\omega_{24}) = S(\omega_{13})$ we recover (\ref{ssa1}).
\end{proof}

It is well-known that, in contrast to classical entropy, quantum entropy does not necessarily increase in the number of parties: $S(\omega_1) \not\le S(\omega_{12})$. Proposition~\ref{pro1} is therefore trivial in the classical case but bounds the sum of entropies of single party quantum states by that of entropies of joint extensions. Our main result is contained in the following bounds:

\begin{proposition}
\label{pro2}
Consider a state $\rho$, a quantum operation $\Phi$ and the image of $\rho$ under $\Phi$: $\rho' = \Phi(\rho) = \sum_i K_i\rho K_i^\dagger$. The complementary channel produces a correlation matrix $\sigma = \tilde\Phi(\rho)$ with elements specified in~(\ref{sigma}). Define the probability vector with entries $q_i := \tr \rho K_i^\dagger K_i$ and quantum states $\rho'_i := K_i\rho K_i^\dagger/q_i$ so that $\rho' = \sum_i q_i\rho'_i$. Then \\[6pt]
a) the Holevo quantity is bounded by the exchange entropy:
\begin{equation}
\chi \bigl( \{q_i, \rho'_i \} \bigr) \le S(\sigma) \le H(\{q_i\}) \enskip\text{and}
\label{propeq1}
\end{equation}
b) the average entropy is bounded by the entropy of the initial state,
\begin{equation}
\sum_i q_i S(\rho'_i) \le S(\rho).
\label{propeq2}
\end{equation}
\end{proposition}

\begin{proof}
a) The rightmost inequality: $S(\sigma) \le H(\{q_i\})$ is a direct consequence of the majorization theorem which says that for any state $S(\sigma) \le S({\rm diag}(\sigma))$, see e.g.~\cite{BZ06}. To prove the left inequality 
consider the isometry $F |\phi\> := \sum_i |i\> \otimes |i\> \otimes K_i|\phi\>$ and the three-partite quantum state
\begin{equation}
\omega_{123} := F \rho F^\dagger = \sum_{ij} |i\>\<j| \otimes |i\>\<j| \otimes K_i \rho K_j^\dagger.
\label{om123}
\end{equation}
It is convenient to introduce the notation $A_{ij} := K_i \rho K_j^\dagger$, so that $q_i = \tr A_{ii}$ and $\rho'_i = A_{ii}/q_i$. One checks that
\begin{align}
&S(\omega_{12}) = S(\sigma) \nonumber \\
&S(\omega_3) = S(\sum_i q_i \rho'_i) \enskip\text{and} \nonumber \\
&- \sum_i q_i S(\rho'_i) = \sum_i \tr A_{ii}\ln A_{ii} - \sum_i \tr A_{ii} \ln \tr A_{ii} \nonumber \\
&\quad = S(\omega_1) - S(\omega_{23}).
\label{average}
\end{align}
Substituting these expressions in~(\ref{ssa2}) yields the first inequality in~(\ref{propeq1}).

b) Since the transformation $F$ in (\ref{om123}) is an isometry the three-partite state $\omega_{123}$ has the same spectrum as $\rho$ up to multiplicities of zero. Hence, $\omega_{123}$ and $\rho$ have the same entropy.
The equality~(\ref{average}) and the Araki--Lieb inequality $S(\omega_{1}) - S(\omega_{23}) \le S(\omega_{123})$
then yield (\ref{propeq2}).
\end{proof}

The bounds in proposition~\ref{pro2} are universal, they hold for any quantum operation $\Phi$ and any initial state $\rho$. We analyze here some of their consequences. The inequality~(\ref{propeq1}) is saturated for orthogonal Kraus operators, $\tr K_i^\dagger K_j = \delta_{ij} K_i = \delta_{ij} K_i^\dagger$ which form a projective von~Neumann measurement. In this case all output states $\rho'_i$ are pure, so $\sum_i q_i S(\rho_i) = 0$. The state $\rho'$ is a mixture of pure and orthogonal states with probabilities $q_i =\tr \rho K_i^\dagger K_i$. The correlation matrix $\sigma$ is then diagonal and $S(\rho') = S(\sigma)$.

Note also that inequality~(\ref{propeq2}) differs from $\bar S = \sum_i q_i S(\rho'_i) \le S(\rho')$, which is implied by the concavity of entropy. For a bistochastic map $\Phi$ the entropy does not decrease, so in this case we may write $\bar S \le S(\rho) \le S(\rho')$.

The Jamio{\l}kowski isomorphism represents of a quantum map $\Phi$ acting on an $N$-level system by a density matrix $\sigma_\Phi$ on an extended space: $\sigma_\Phi = \Phi \otimes \id (|\phi^+\>\<\phi^+|)$, where $|\phi^+\> = \frac{1}{\sqrt{N}} \sum_i |i\> \otimes |i\>$ is a maximally entangled state.
The degree of non-unitarity of an operation $\Phi$ can be quantified by its entropy~\cite{BZ06}, defined as the entropy of the corresponding state  $S(\Phi) := S(\sigma_\Phi)$. If the initial state $\rho$ is maximally mixed, then the exchange entropy $S(\sigma)$ is equal to the entropy of the map~\cite{roga}. Proposition~\ref{pro2} yields now a simple interpretation of the entropy of a map: it is an upper bound for the Holevo quantity~(\ref{holev}) for a transformation of the maximally mixed state $\rho_* = \idty /N$. Furthermore, the entropy of a map is an upper bound for the Holevo quantity associated with ensembles of Kraus maps acting on the mixed state $\rho_*$
\begin{equation}
\max_{\{K_i\}}\  \chi\Big(
\Bigl\{ \tr K_i K_i^\dagger/N, \frac{K_i  K_i^\dagger}{\tr K_i K_i^\dagger} \Bigr\} \Bigl) \;\le\; S(\Phi),
\label{entop}
\end{equation}
where the maximum is taken over sets of Kraus operators that realize the same quantum operation: $\Phi(\rho) = \sum_i K_i\rho K_i^\dagger$.

\section{Coherent Information}
\label{s4}

Investigating the entropy transfer induced by a quantum map $\Phi$ that sends a state $\rho$ to $\rho' = \Phi(\rho) = \sum_i K_i\rho K_i^\dagger$, Lindblad~\cite{lindblad} proved the inequality 
\begin{equation}
S(\rho') - S(\sigma) \le S(\rho) \le S(\rho') + S(\sigma).
\label{lin}
\end{equation}
The correlation matrix $\sigma$ is defined in~(\ref{sigma}) and the proof is based on sub-additivity of entropy and on the Araki-Lieb triangle inequality. The difference of entropies, $I_{\text{coh}} := S(\Phi(\rho)) - S(\sigma)$ is called {\sl coherent information}. Linblad's inequality states that $I_{\text{coh}} \le S(\rho)$. We are now in position to refine this bound.

\begin{proposition}
Consider a state $\rho$ and quantum operations $\Phi_1$ and $\Phi_2$, with  $\Phi_1(\rho) = \sum_i K_i \rho K_i^\dagger$ and a quantum ensemble $\{q_i,\rho'_i\}$ where $q_i := \tr \rho K_i^\dagger K_i$ and $\rho'_i := K_i \rho K_i^\dagger/q_i$. Then \\[6pt]
a) The coherent information for the quantum operation $\Phi_1$ is bounded by
\begin{equation}
I_{\text{coh}}(\Phi_1) \le \sum_i q_i S(\rho'_i) \le S(\rho).
\label{icoh}
\end{equation}
b) The coherent information for the concatenation $\Phi_2 \circ \Phi_1$ is bounded by
\begin{equation}
I_{\text{coh}}(\Phi_2 \circ \Phi_1) \le \sum_i p_i S\bigl(\Phi_2(\rho'_i)\bigr).
\label{icoh2}
\end{equation}
\end{proposition}

\begin{proof}
Relation~(\ref{icoh}) is a direct consequence of proposition~\ref{pro2}, as these inequalities are obtained by combining~(\ref{propeq1}) and~(\ref{propeq2}). To show~(\ref{icoh2}) we consider the four-partite state
\begin{equation}
\omega'_{1234} := \sum_{ijk\ell} |i\>\<j| \otimes |k\>\<\ell| \otimes |k\>\<\ell| \otimes L_iK_k\rho K^\dagger_\ell L^\dagger_j,
\label{omega1234}
\end{equation}
where $\Phi_2(\rho) = \sum_i L_i \rho L^\dagger_i$. Consider the strong sub-additivity relation
\begin{equation}
S(\omega'_4) + S(\omega'_{13}) \le S(\omega'_{123}) + S(\omega'_{24}).
\end{equation}
The left-hand side inequality~(\ref{icoh2}) which we want to prove can be rewritten as
\begin{equation}
S(\omega'_4) + S(\omega'_3) \le S(\omega'_{123}) + S(\omega'_{24}).
\end{equation}
Therefore, it is sufficient to prove that $S(\omega'_3) \le S(\omega'_{13})$. As the matrix $\omega'_3$ is diagonal
and consists of the traces of the blocks of the block diagonal matrix $\omega'_{13}$, $\omega'_{13}$ is more mixed than $\omega'_3$ and has therefore larger entropy.
\end{proof}

\section{The Jensen-Shannon Divergence}
\label{s5}

Let us now consider the static case in~Fig.~1b: a quantum ensemble $\{q_i,\rho_i\}$ which determines the average state $\bar \rho := \sum_i q_i \rho_i$. For an ensemble of classical measures, $\{q_i,\mu_i\}$ one defines the {\sl generalized Jensen-Shannon divergence}~(JSD) by
\begin{equation}
JS(\{q_i,\mu_i\}) :=  H\Bigl(\sum_i q_i \mu_i \Bigr) - \sum_i q_i H(\mu_i),
\end{equation}
which is an exact classical analogue of the expression~(\ref{holev}). Hence the Holevo quantity $\chi$ is often called the Quantum Jensen-Shannon divergence~(QJSD)~\cite{topsoe,MLP95,briet}.

It is intuitively clear that the results in Section~\ref{s3} may be used to derive upper bounds for QJSD, although the map $\Phi$ in Proposition~\ref{pro2} and its Kraus form are not specified here. For any initial state $\rho$ and an arbitrary Kraus operator $K_i$ we consider the polar decomposition $K_i\rho^{1/2} = X_iU_i$. Here $X_i$ is a Hermitian matrix and $U_i$ is unitary. Note that $X_i^2 = K_i \rho K_i^\dagger$, and this is equal to $q_i\rho'_i$. Therefore $K_i \sqrt{\rho} = \sqrt{q_i\rho'_i} U_i$, so the elements of the correlation matrix~(\ref{sigma}) read:
\begin{equation}
\sigma_{ij} = \tr K_i \rho K_j^\dagger K_i = \sqrt{q_i  q_j }\; \tr \sqrt{\rho_i} U_i U_j^\dagger\sqrt{\rho_j}.
\label{sigma2}
\end{equation}
In this way we arrive at

\begin{corollary}
Consider a quantum ensemble $\{q_i,\rho_i\}$ and a collection of unitary matrices $\{U_i\}$ and construct the correlation matrix $\sigma$ as in~(\ref{sigma2}). Then $\chi (\{q_i,\rho_i\}) \le S(\sigma)$.
\end{corollary}

As a simple application consider the case of an ensemble with 2 elements. To obtain the lowest upper bound for QJSD
we need to minimize the entropy of the correlation matrix~(\ref{sigma2}) over sets of unitaries $\{U_1,U_2\}$. This is equivalent to finding the POVM which minimizes $S(\sigma)$ among all measurements which result in the same ensemble of output states.

\begin{lemma}
Consider two density matrices $\rho_1$ and $\rho_2$ occurring with probabilities $(\lambda, 1-\lambda)$.
The smallest entropy of the correlation matrix~(\ref{sigma2}) over unitaries $U_1$ and $U_2$ is achieved for the matrix
\begin{equation}
\sigma_{\lambda} = \begin{pmatrix}
\lambda & \sqrt{\lambda(1-\lambda)} \sqrt F \\
\sqrt{\lambda(1-\lambda)}\sqrt F & 1-\lambda
\end{pmatrix},
\label{sigmamin}
\end{equation}
where $\sqrt F$ is the root fidelity~\cite{Jo94}: $\sqrt F = \tr\sqrt{ \rho_1^\frac{1}{2} \rho_2 \rho_1^\frac{1}{2} }$.
\end{lemma}

\begin{proof}
Given $\lambda$, $\rho_1$, and $\rho_2$ the entropy $S(\sigma)$ is minimal, if the absolute value of the off-diagonal element $\sigma_{12}$ is maximal. As $\abs{\tr A\,B} \le \norm a\, \tr \abs B$ we have the upper bound
\begin{equation}
|\tr \rho_2^\frac{1}{2} \rho_1^\frac{1}{2} U_1 U_2^\dagger| \le \tr |\rho_2^\frac{1}{2} \rho_1^\frac{1}{2}| = \sqrt F.
\end{equation}
Moreover the inequality is saturated by choosing for $U_1 U_2^\dagger$ the adjoint of the unitary of the polar decomposition of $\rho_2^\frac{1}{2} \rho_1^\frac{1}{2}$.
\end{proof}

Let us now set $\lambda = 1-\lambda = \tfrac{1}{2}$. In this case the quantum Jensen-Shannon divergence can be written as $\r{QJS}(\rho_1,\rho_2) = \tfrac{1}{2}\, [S(\rho_1 \Vert \bar \rho) + S(\rho_2 \Vert \bar \rho)]$, where $\bar \rho = \tfrac{1}{2}\, (\rho_1 + \rho_2)$ and  $S(\rho_1 \Vert \rho_2) := \tr \rho_1(\ln \rho_1 - \ln \rho_2)$
is the quantum relative entropy~\cite{nielsen}. The spectrum of the correlation matrix~(\ref{sigmamin}) is
$(\mu, 1-\mu)$ with $\mu = \tfrac{1}{2}\, (1 - \sqrt F)$. Hence we get an explicit formula for the exchange entropy $S(\sigma)$ and the universal bound for QJSD 
\begin{equation}
\r{QJS}(\rho_1,\rho_2) \le H_2 \Bigl( \tfrac{1}{2}\, \bigl( 1 - \sqrt{F(\rho_1,\rho_2)} \bigr) \Bigr)
\label{bound}
\end{equation}
where $H_2(x) := -x\ln x - (1-x) \ln{(1-x)}$ is the Shannon entropy of a probability vector of size 2.

The right-hand side of~(\ref{bound}) can be used to characterize closeness between quantum states. Although the entropy $H_2$ does not obey the triangle inequality its square root does. Such an {\sl entropic distance} was advocated by Lamberti et al.~\ \cite{LPS09} as a natural metric in the space of quantum states:
\begin{equation}
D_E(\rho_1,\rho_2) := \sqrt{ H_2 \bigl( \tfrac{1}{2}\, (1 - \sqrt{F(\rho_1,\rho_2)}) \bigr)}.
\label{entropic}
\end{equation}
To show that $D_E$ is a distance one may use  the {\sl Bures distance}, $D_B(\rho_1,\rho_2) = \sqrt{2 - 2\sqrt{F(\rho_1,\rho_2)}}$. Since both quantities are functions of fidelity, one can write the entropic distance as a function of the Bures distance $D_E(D_B) = \sqrt{H_2(D_B^2/4)}$. As the second derivative of $D_E(D_B)$ is negative, this function is concave, so $D_E$ satisfies the axioms of a distance.

Turning now to the classical case, i.e., diagonal density matrices, the root fidelity reduces to the {\sl Bhattacharyya coefficient} $B(P,Q) = \sum_i \sqrt{p_iq_i}$ while the Bures distance is equivalent to the {\sl Hellinger distance}~\cite{BZ06} $D_H(P,Q) = \sqrt{\sum_i (\sqrt{p_i}-\sqrt{q_i})^2}$. The entropic distance $D_E$ between two classical states is then a concave function of their Hellinger distance, $D_E(P,Q) = \sqrt{H_2(D_H^2(P,Q)/4)}$. Although JSD does not satisfy the triangle inequality, its square root does~\cite{ES03} and is called the transmission distance~\cite{briet}. Inequality~(\ref{bound}) implies thus the following relation between the transmission  distance $D_T$ and the entropic distance $D_E$ used in~\cite{LPS09},
\begin{equation}
D_T(P,Q):= \sqrt{\r{DJS}(P,Q)} \le D_E(P,Q),
\label{entbound}
\end{equation}
which is illustrated in Fig.~\ref{fig:fig2}.

\begin{figure}[htbp]
\centering
\includegraphics[width=0.49\textwidth]{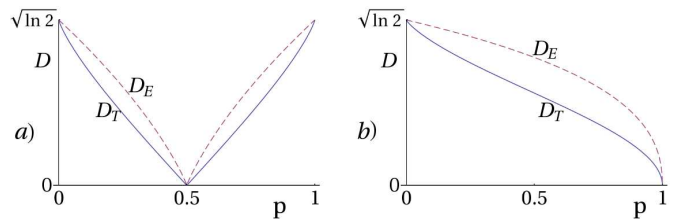}
\caption{The transmission distance $D_T$ (solid line) is bounded from above by the entropic distance $D_E$ (dashed line) \\ 
a) for pairs of classical states $P = (p,1-p)$ and $Q=(1-p,p)$ \\
b) for pairs of classical states $P = (p,1-p)$ and $Q=(1,0)$.}
\label{fig:fig2}
\end{figure}

\section{Conclusion}

In this work we showed that the Holevo quantity $\chi$ is bounded by the exchange information $S(\sigma)$
and we analyzed some consequences of this result. In particular, we showed that the transmission distance $D_T$  between classical states is bounded by their entropic distance $D_E$.

\section*{Acknowledgements}

It is a pleasure to thank R.~Alicki and M.P.H.~Horodecki for fruitful discussions and N.~Datta, P.~Herremo{\"e}s,  P.~Lamberti, and F.~Tops{\o}e  for helpful correspondence. This work was supported by the grant number DFG-SFB/38/2007 of the Polish Ministry of Science and Higher Education and by the Belgian Interuniversity Attraction Poles Programme P6/02.

\end{document}